\documentclass{elsarticle}

\usepackage{graphicx}
\usepackage{amssymb,amsmath,amsthm}

\newcommand{\D}{{\mathrm{d}}}
\newtheorem{theorem}{Theorem}
\newtheorem{lemma}{Lemma}

\newtheorem{corollary}{Corollary}

\begin{document}

\begin{frontmatter}

\title{Local equivalence of reversible \\ and general Markov kinetics}

\author{Alexander N. Gorban}

 \ead{ag153@le.ac.uk}
\address{Department of Mathematics, University of Leicester, Leicester, LE1 7RH, UK}

\begin{abstract}
We consider continuous--time  Markov kinetics with a finite number of states and a
positive equilibrium $P^*$. This class of systems is significantly wider than the systems
with detailed balance. Nevertheless, we demonstrate that for an arbitrary probability
distribution $P$ and a general system there exists a system with detailed balance and the
same equilibrium that has the same velocity $\D P / \D t$ at point $P$. The results are
extended to nonlinear systems with the generalized mass action law.\end{abstract}

\date{}

\end{frontmatter}

\section{Introduction}
\subsection{Detailed balance and beyond \label{Overview}}

The principle of detailed balance is one of the most celebrated results in kinetics. A
kinetic system is represented as a mixture of independent elementary processes
(collisions or elementary reactions, for example). Due to the principle of detailed
balance, {\em at equilibrium, each elementary process should be equilibrated by its
reverse process.} Kinetics is decomposed into pairs of mutually inverse processes and in
many problems we can consider these pairs separately.

We study relations between the systems with and without detailed balance. In this Section, we
briefly overview the main results of the work. Then, in Sec.~\ref{History} we review the
history of the problem. We prove the local equivalence theorem for the Markov processes
in Sec.~\ref{Markov} and give there the simple examples. The nonlinear generalizations
are presented in  Sec.~\ref{Nonlin}.

In Sec.~\ref{Markov} we start from the first order kinetics without the detailed balance assumption.
The general first order kinetic equation has the form:
\begin{equation}\label{MAsterEq0}
\frac{\D p_i}{\D t}= \sum_{j, \, j\neq i} (q_{ij}p_j-q_{ji}p_i) \, ,
\end{equation}
where $q_{ij}$ ($i,j=1,\ldots, n$, $i\neq j$) are non-negative. This system of
equations (master equations or Kolmogorov's equations) describes dynamics of
non-negative variables $p_i$ ($i=1,\ldots, n$). These variables may be
considered as probabilities (then $\sum_i p_i =1$) or concentrations. For the
corresponding  states or components we use the notation $A_i$. In this
notation, $q_{ij}$ is the {\em rate constant} for transitions $A_j \to A_i$.

Let us assume that system (\ref{MAsterEq0}) has a positive equilibrium
$P^*=(p_i^*)$, $p_i^*>0$:
\begin{equation}\label{MasterEquilibrium}
\sum_{j, \, j\neq i} q_{ij}p^*_j = \left(\sum_{j, \, j\neq i}
q_{ji}\right)p^*_i \;\mbox{ for all } i=1,\ldots n \, .
\end{equation}
This is the so-called {\em balance equation}.

The {\em detailed balance} condition is much stronger. It assumes that the sums in
the left and right hand sides of Eq. (\ref{MasterEquilibrium}) are equal term by term:
\begin{equation}\label{detBal}
q_{ij}p^*_j=q_{ji}p^*_i \;\mbox{ for all } i,j=1,\ldots n, \, i \neq j \,  .
\end{equation}

For the number of states $n>2$, a simple comparison of dimensions demonstrates
that there are much more general systems with the given positive equilibrium
$P^*$ (\ref{MasterEquilibrium}) (dimension is $(n-1)^2$) than the systems with detailed balance with
equilibrium $P^*$ (\ref{detBal}) (dimension is $\frac{n(n-1)}{2}$). Surprisingly, {\em for
every given distribution $P$, the set of possible velocities $\D P / \D t$ for
general Markov kinetics with equilibrium $P^*$ is the same that for Markov
kinetics with detailed balance and the same equilibrium.}
This is the central result of the paper (Theorem~\ref{Theorem1} in Sec.~\ref{Markov}).

We demonstrate this in two steps. First, we use the  representation of  a general Markov chain with a given
positive equilibrium as a combination with non-negative coefficients
of several simple cycles with the same equilibrium.

Secondly, we demonstrate this equivalence for a
simple cycle of transitions with positive constants
\begin{equation}\label{cycle}
A_1 \to A_2 \to \ldots \to A_n \to A_1 \, .
\end{equation}
For the equilibrium $P^*$ the constants of the cycle are $q_{i+1 i}=\kappa
/p_i^*$ (we use here the standard convention about the cyclic numeration,
$n+1=1$).

Thus, if we observe the Markov kinetics at one point then we can not distinguish
general systems from systems with detailed balance because the sets of possible
velocities coincide. In particular, they have the same set of Lyapunov
functions.

Our main results allow us to decompose any Markov kinetics (or generalized mass action
law kinetics with semi-detailed balance) into pairs of mutually inverse elementary
processes with the same equilibrium. If the system does not satisfy the principle of
detailed balance then this decomposition depends on the state. Nevertheless, in some
problems it is still convenient to consider these pairs separately.

In this paper, we give two examples of the application of Theorem~\ref{Theorem1}: the
evaluation of logarithmic decrement  for general Markov chains and a simple proof of the
Morimoto $H$-theorem for all the Csisz\'ar--Morimoto divergencies. We give also the
nonlinear generalization of Theorem~\ref{Theorem1} for the systems which obey the
generalized {\em Mass Action Law} (MAL).

\begin{figure}
\centering{\includegraphics[width=0.5\textwidth]{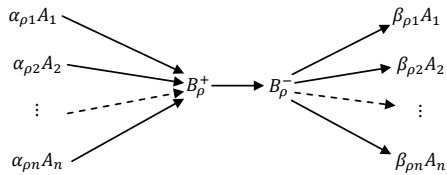}
\caption{\label{Compounds}Elementary reaction with intermediate compounds.}}
\end{figure}

Master equation is a source for many other kinetic equations. In particular, in Sec.~\ref{Nonlin} we consider
complex reactions with intermediate compounds (Fig.~\ref{Compounds}) under two
asymptotic conditions
\begin{itemize}
\item{The compounds $B_j$ are in fast equilibrium with the corresponding input or output
    reagents;}
\item{They exist in very small concentrations compared to other components.}
\end{itemize}
For compounds transitions the first order kinetics is assumed because of small
concentrations of compounds (only first order terms survive). These assumptions allow us
to produce the reaction rates for the overall reaction from Fig.~\ref{Compounds} in the
form of the generalized MAL:
\begin{equation}\label{generalizedMAL0}
r_{\rho}=\varphi_{\rho}\exp\left(\frac{\sum_i\alpha_{\rho i}
\mu_i}{RT}\right)\, , \; \varphi_{\rho} \geq 0 \, .
\end{equation}
where $\mu_i$ is the chemical potential  of the component  $A_i$, $\rho$ is the reaction
number, $\alpha_{\rho i}$ are the input stoichiometric coefficients
(Fig.~\ref{Compounds}). Both $\alpha_{\rho i}$ and $\beta_{\rho i}$ are non-negative
integers. We use notations $\boldsymbol{\alpha}_{\rho}$ and $\boldsymbol{\beta}_{\rho}$
for vectors wit coordinates $\alpha_{\rho i}$ and $\beta_{\rho i}$. The positive
functions $\varphi_{\rho}$ are called the {\em kinetic factors} whereas
$\exp\left({\sum_i\alpha_{\rho i} \mu_i}/RT\right)$ are the {\em Boltzmann factors}.

The balance condition of the first order kinetics of compounds (\ref{MasterEquilibrium}) transforms in
the semi-detailed balance condition (that is known
also as the complex or the cyclic balance condition):
$$
\sum_{\rho, \,\boldsymbol{\alpha}_{\rho}=\boldsymbol{\nu}} \varphi_{\rho}\equiv
\sum_{\rho, \,\boldsymbol{\beta}_{\rho}=\boldsymbol{\nu}} \varphi_{\rho}
$$
for any vector $\boldsymbol{\nu}$ from the set of all vectors
$\{\boldsymbol{\alpha}_{\rho}, \boldsymbol{\beta}_{\rho}\}$. Kinetics with the
generalized MAL and the semi-detailed balance conditions give the natural nonlinear
generalizations of Markov processes. In particular, the entropy production for these
systems at any nonequilibrium state is positive.

If we assume for the Markov kinetics of compounds that the positive
equilibrium is the point of detailed balance (\ref{detBal}) then the kinetic factors
$\varphi_{\rho}$ satisfy the stronger condition of detailed balance:
$$
\varphi_{\rho}^+\equiv \varphi_{\rho}^- \; \mbox{ for all } \rho \, ,
$$
where $\varphi_{\rho}^+$ is the kinetic factor for the direct reaction and
$\varphi_{\rho}^-$ is the kinetic factor for the reverse reaction.

The class of systems with semi-detailed balance is much wider than the class of systems with detailed
balance. Nevertheless, locally they coincide: {\em the set of  possible velocities for systems with semi-detailed balance coincide
with the set of possible velocities for the systems with detailed
balance for given thermodynamic functions and any given state}  (Sec.~\ref{Nonlin}).

The systems with generalized MAL and semi-detailed balance are the nonlinear analogs of
the Markov processes and the local equivalence of the generalized MAL systems with
detailed and semi-detailed balance is the analog and a consequence of
Theorem~\ref{Theorem1}.

\subsection{A bit of history \label{History}}

In 1872, Boltzmann introduced the principle of detailed balance for collisions
and used it to prove his $H$-theorem \cite{Boltzmann1872}.
Boltzmann's proof of the positivity of entropy production for systems with detailed
balance is very transparent because it is sufficient to prove this positivity just for a
couple of mutually inverse elementary processes.

In 1887, Lorentz \cite{Lorentz1887} objected Boltzmann: he insisted that some collisions
of polyatomic molecules do not have reverse collisions and cannot satisfy the principle
of detailed balance. Immediately, Boltzmann realized that there exists a much weaker
condition sufficient for the $H$-theorem \cite{Boltzmann1887}. In 1981, it was proven
that the Lorentz objections are wrong and the principle of detailed balance is valid for
polyatomic molecules \cite{CercignaniLampis1981}. This is not very surprising because the
detailed balance follows from microreversibility (or $T$-invariance of the fundamental
equations of mechanics or quantum mechanics). Nevertheless, the Boltzmann discovery is
valuable by itself and is used for many kinetic equations.

This condition was rediscovered several times. It is known as the semi-detailed balance
condition, the cyclic balance condition or the complex balance condition. In 1952,
Stueckelberg proposed a proof of the semi-detailed balance condition for the Boltzmann
equation \cite{Stueckelberg1952}. His proof is based on the Markov model of elementary
events. Recently, the Stueckelberg approach was extended to prove the semi-detailed
balance condition for the generalized MAL kinetics \cite{GorbanShahzad2011}.

The complex balance condition for chemical kinetics was introduced by Horn and Jackson in
1972 \cite{HornJackson1972} independently of Boltzmann's work. Now it is used for
mathematical modeling in chemical kinetics and engineering \cite{SzederkHangos2011}.
Boltzmann's idea about cyclic balance developed in physical kinetics was independently
rediscovered in the theory of Markov processes and it is proved that any recurrent Markov
process can be decomposed into directed cycles \cite{Kalpazidou2006}.

The principle of detailed balance was crucially important in the development of the
Metropolis--Hastings and other Markov chain Monte Carlo algorithms from the very beginning
\cite{Metropolis1953}. Technically, it is much easier to use the detailed balance conditions
(\ref{detBal}) than to follow more intricate balance conditions. Detailed balance was
considered as a necessary condition for construction of Monte Carlo algorithms as a
``systematic design principle" \cite{Katsoulakis2003}. It was demonstrated that it helps
to reduce the uncertainty of some observables in stochastic numerics \cite{Noe2008}.

Nevertheless, there are many examples of efficient Monte Carlo computations without
detailed balance. Sometimes computational models without detailed balance are constructed
because of the physical nature of the systems. For example, the models of inelastic
processes in particle physics \cite{Grassberger1979} or in granular media \cite{Dean2003}
may violate the principle of detailed balance. The general theories of stochastic
cellular automata with Gibbsian equilibria but without compulsory detailed balance were
developed \cite{Marroquin1991}. It is widely recognized that the balance equation (not
the detailed balance) is necessary and sufficient condition for invariance (stationarity)
of the desired equilibrium distribution. Under some more technical irreducibility
conditions, the Monte Carlo simulations will converge to this equilibrium
\cite{Manousiothakis1999,Athenesa2007}.

The interplay between reversible (with detailed balance) and irreversible Markov chains
is non-trivial and important for many applications. Recently, it was demonstrated that
the local deformation of the reversible Markov processes into irreversible ones helps to
create efficient computational Monte Carlo algorithms \cite{Turitsyn2011}.

Much efforts were applied to verification of microscopic reversibility and its
consequences, the detailed balance conditions  and the Onsager reciprocal relations, in many experimental
systems \cite{Miller1960,Thornton1971,Driller1979,Rettner1989,YablonskyGor_atal2011}. To
check experimentally the detailed balance conditions it is necessary to deal with a
complex reaction that can be formally equilibrated without detailed balance. If not, then
one tests not the detailed balance but just the equilibrium condition as it was mentioned
in \cite{Henley1959}.

The detailed balance conditions are very natural and appealing. They simplify many
computations and proofs. Thus, they are used in many applications and models. For
some physical and chemical systems, detailed balance has a solid background, the
$T$-invariance of the fundamental equations, but it is also used beyond the proven
microreversibility. When modelling with detailed balance meets some difficulty then
the problem about relations between reversible and irreversible systems arises again and
detailed balance is substituted by more general conditions. Nevertheless, the
convenience, beauty and some intrinsic benefits of the phenomenon, have forced researchers to return to detailed
balance if it is possible without a conflict with reality. There are many examples of
this ``pendulum" in scientific publications: accept detailed balance -- criticize detailed balance --
go to more general conditions -- realize the benefits of detailed balance -- return to detailed balance -- ...

At the same time, the consequences of the principle of detailed balance are extended to
the situations where it was not used before. Thus, recently the multiscale limit of the
systems of reversible reactions was studied when some of the equilibrium concentrations tend
to zero. The extended principle of detailed balance was proved for the systems with some
irreversible reaction \cite{GorbanYablonsky2011,GorbanMirYab2012}.

In Theorem~\ref{Theorem1}, we compare the sets of possible velocities, $\D P/\D t$, for
two classes of systems: (i) general first order kinetics (\ref{MAsterEq0}) with the given
positive equilibrium $P^*$ and (ii) first order kinetics with detailed balance
(\ref{detBal}) and the same positive equilibrium.

Understanding the structure of the sets of possible velocities can
provide additional information about attainable states of the system
which is helpful in the modelling context. There are many other reasons too.
It is known that different types of kinetics data bear different degrees of
reliability. It would therefore  be very attractive to study the
consequences of the information of each level of reliability separately \cite{GorbanTree2012}.
For example,  uncertainty about equilibria in the system is usually significantly
lower than that of the reaction rate constants. The value of the equilibrium gives
us some information about dynamics: the set of possible velocities $\D P/\D t$
is not arbitrarily wide at a given state and for given equilibrium. For the
systems with detailed balance, this set is a polyhedral cone which allows a simple
description (Sec.~\ref{Quasichem}). Due to Theorem~\ref{Theorem1}, however, this cone is also the
set of possible velocities for the general master equation. Therefore,
to distinguish the detailed balance systems from the general ones we have to involve data about
$\D P/\D t$ for several significantly distant distributions $P$.

Another example is to employ the knowledge of the sets of possible velocities
for estimating attainable regions for kinetics. The idea to use the equilibrium information to estimate the
attainable sets in kinetics was proposed in 1964 by Horn \cite{Horn1964} and developed
further in chemical kinetics and chemical engineering
\cite{Gorban1979,obh,Glasser1987,Hildebrandt1990,GorbKagan2006} (for more detailed review
see \cite{GorbanTree2012}). If the sets of possible velocities coincide then the
estimated attainable regions coincide too. The knowledge of the sets of possible
velocities is also important for the analysis of observability, identifiability and controllability of the
systems.

\section{Local equivalence of  general Markov systems and systems with detailed balance \label{Markov}}

\subsection{Global decomposition into cycles, local decomposition into steps, and the equivalence theorem}

Let us start from master equation (\ref{MAsterEq0}). The coefficient $q_{ij}$ is the {\em
rate constant} for transitions $A_j \to A_i$. Any set of non-negative coefficients
$q_{ij}$ ($i\neq j$) corresponds to a master equation. Therefore, the set of all master
equations (\ref{MAsterEq0}) may be considered as the non-negative orthant in
$\mathbb{R}_+^{n(n-1)} \subset \mathbb{R}^{n(n-1)}$. (The non-negative orthant is the set
of all vectors with only non-negative components.)

We assume that system (\ref{MAsterEq0}) has a positive equilibrium $P^*=(p_i^*)$,
$p_i^*>0$ and the balance condition (\ref{MasterEquilibrium}) holds. The sum of these $n$
{\em balance conditions} is a trivial identity. Let us delete any single equation from
(\ref{MasterEquilibrium}), for example, the last one (for $i=n$). Each of the remaining
equations  includes the variable $q_{in}$ which is not present in other equations
($i=1,\ldots ,n-1$). Therefore, for given positive $P^*$, there are $n-1$ independent
conditions on $q_{ij}$ ($i,j=1,\ldots, n$, $i\neq j$) in (\ref{MasterEquilibrium}).

Thus, the balance conditions (\ref{MasterEquilibrium}) define for a given positive
equilibrium a $(n-1)^2$-dimensional linear subspace $L^*$ in the $n(n-1)$ dimensional
space of $q_{ij}$ ($i\neq j$).  A vector of positive coefficients $q_{ij}^*=1/p_j^*$
satisfies (\ref{MasterEquilibrium}) and belongs to $L^*$. This vector belongs to the
interior of the non-negative orthant $\mathbb{R}_+^{n(n-1)}$. Therefore, the intersection
$\mathbb{R}_+^{n(n-1)} \cap L^*$ includes a vicinity of $q_{ij}^*$ in $L^*$ and is a
$(n-1)^2$-dimensional cone. Thus, the non-negative solutions of (\ref{MasterEquilibrium})
form a $(n-1)^2$-dimensional closed cone in $\mathbb{R}_+^{n(n-1)}$.

The systems with detailed balance (\ref{detBal}) for a given positive equilibrium form a smaller cone.
Under these conditions, there are only $\frac{n(n-1)}{2}$ independent coefficients among
$n(n-1)$ numbers $q_{ij}$. For example, we can arbitrarily select $q_{ij} \geq 0$ for
$i>j$ and then take $q_{ij} = q_{ji}\frac{p_i^*}{p^*_j}$ for $i<j$. So,  for given $P^*$,
the cone of the detailed balance systems (\ref{detBal}) can be considered as a non-negative orthant in
$\mathbb{R}^{\frac{n(n-1)}{2}}$ embedded in $\mathbb{R}_+^{n(n-1)}$.

If the balance condition (\ref{MasterEquilibrium}) holds then system
(\ref{MAsterEq0}) may be rewritten in a convenient  equivalent form:
\begin{equation}\label{MAsterEq1}
\frac{\D p_i}{\D t}= \sum_{j, \, j\neq i}
q_{ij}p^*_j\left(\frac{p_j}{p_j^*}-\frac{p_i}{p_i^*}\right) \, .
\end{equation}
With this form of master equation, it is straightforward to calculate the time
derivative of the quadratic divergence, a weighted $l_2$ distance between $P$
and $P^*$, $H_2(P\|P^*)=\sum_i\frac{(p_i-p^*_i)^2}{p^*_i}$:
\begin{equation}\label{QuadLyap}
\frac{\D H_2(P\|P^*)}{\D t}=-\sum_{i,j, \, j\neq i}q_{ij}p^*_j\left(\frac{p_i}{p^*_i} - \frac{p_j}{p^*_j}\right)^2 \leq 0\, .
\end{equation}
This time derivative is {\em strictly negative} if for a transition $A_j \to
A_i$ the rate constant is positive, $q_{ij}>0$, and $\frac{p_i}{p^*_i} \neq
\frac{p_j}{p^*_j}$. Hence, if the state $P$ is not an equilibrium (i.e., the
right hand side in (\ref{MAsterEq1}) is not zero) then $\frac{\D
H_2(P\|P^*)}{\D t}<0$.

Let us introduce the following notation for a given number of states $n$:
\begin{itemize}
\item{$\mathcal{Q}^n_{\rm B}(P^*)$ is the cone of the vectors of non-negative
    coefficients $q_{ij}$ ($i\neq j$) which satisfy the balance conditions
    (\ref{MasterEquilibrium}), that is, the set
    of all Markov processes with the equilibrium distribution $P^*$;}
\item{$\mathcal{Q}^n_{\rm DB}(P^*)$ is the cone of the vectors of non-negative
    coefficients $q_{ij}$ ($i\neq j$) which satisfy the detailed  balance conditions
    (\ref{detBal}), that is,
    the set of all Markov processes with detailed balance and the equilibrium distribution $P^*$.}
\end{itemize}

If a system satisfies the detailed balance condition (\ref{detBal}) then the balance
condition (\ref{MasterEquilibrium}) holds too: it holds term by term, even without
summation. Therefore, $\mathcal{Q}^n_{\rm DB}(P^*) \subset \mathcal{Q}^n_{\rm B}(P^*)$.
Moreover, comparing dimension we find that this inclusion is strong. Indeed, $\dim
\mathcal{Q}^n_{\rm B}(P^*)=(n-1)^2$, $\dim \mathcal{Q}^n_{\rm DB}(P^*)=\frac{n(n-1)}{2}$.
If $n>2$ then $\frac{n(n-1)}{2} < (n-1)^2$, hence,
\begin{equation}\label{inclusion1}
\mathcal{Q}^n_{\rm DB}(P^*) \subsetneqq \mathcal{Q}^n_{\rm B}(P^*) \subsetneqq \mathbb{R}_+^{n(n-1)} \, .
\end{equation}
The cone of the systems with detailed balance is, in some sense, much smaller than the
cone of the systems with the balance condition: the difference between their dimensions
is $\frac{(n-1)(n-2)}{2}$.

Now, let us consider the right hand side vector fields of the systems
(\ref{MAsterEq0}) at the point $P\neq P^*$. For each cone of the coefficients
$q_{ij}$ the vectors of the possible velocities, $\D P/ \D t$, also form a
cone. Let us introduce the following notation:
\begin{itemize}
\item{$\mathbf{Q}^n_{\rm B}(P,P^*)$ is the cone of the possible velocities, $\D P/ \D
    t$, at the point $P$ for $(q_{ij}) \in \mathcal{Q}^n_{\rm B}(P^*)$;}
\item{$\mathbf{Q}^n_{\rm DB}(P,P^*)$ is the cone of of the possible velocities, $\D
    P/ \D t$, at point $P$ for $(q_{ij}) \in \mathcal{Q}^n_{\rm DB}(P^*)$.}
\end{itemize}
$\mathbf{Q}^n_{\rm B}(P,P^*)$ is the cone of all possible velocities for Markov kinetics
at the point $P$ if the equilibrium is $P^*$. $\mathbf{Q}^n_{\rm DB}(P,P^*)$ is the cone
of these velocities for Markov kinetics with detailed balance.

Surprisingly, {\em for every given distribution $P$, the set of possible velocities $\D P
/ \D t$ for general Markov kinetics with equilibrium $P^*$ is the same that for Markov
kinetics with detailed balance and the same equilibrium.} The following theorem is the
central result of this work.
\begin{theorem}\label{Theorem1}$\mathbf{Q}^n_{\rm B}(P,P^*)=\mathbf{Q}^n_{\rm DB}(P,P^*)$\end{theorem}

This means that for every first order kinetic equation (\ref{MAsterEq0}) with a
given positive equilibrium $P^*$ and for every point $P\neq P^*$ there exists a
first order kinetic equation with detailed balance and  equilibrium $P^*$ that
has the same velocity at $P$. At this point the right hand sides of the kinetic
equations coincide.

Therefore, if we observe the Markov kinetics at one point then we can never
distinguish general systems from systems with detailed balance. In particular,
they have the same set of Lyapunov functions:
\begin{corollary}\label{Corollary1} If for a function
$H(P,P^*)$, $\D H / \D t \leq 0$ for any system (\ref{MAsterEq0}) with
equilibrium $P^*$ and detailed balance then $\D H / \D t \leq 0$ for any system
(\ref{MAsterEq0}) with equilibrium $P^*$.\end{corollary}

The system with detailed balance has $n(n-1)/2$ dimensions available to match a single
$n$-dimensional velocity vector. Therefore, it is not surprising that the cones
$\mathbf{Q}^n_{\rm DB}(P,P^*)$ has non-empty interior ({\em solid cones}). But this is
not enough to cover any $n$-dimensional velocity vector using {\em non-negative}
coefficients $q_{ij}$. This non-negativity condition defines the borders of the cones and
we can a priori just state the inclusion $\mathbf{Q}^n_{\rm DB}(P,P^*) \subseteq
\mathbf{Q}^n_{\rm B}(P,P^*)$. The calculation of dimension does not give a hint about
coincidence of these cones.

The proof of Theorem 1 is constructed in two steps. First, we prove that for
every $P^*$ and $P$ the cone of possible velocities $\mathbf{Q}^n_{\rm
B}(P,P^*)$ is the convex hull of the velocities at point $P$ of the simple
cyclic schemes, $A_{i_1} \to \ldots \to A_{i_k} \to A_{i_1}$ ($k\leq n$ and all
the numbers $i_1, \ldots, i_k$ are different), with the same equilibrium $P^*$.
Secondly, we prove that it is sufficient to take $k=2$.

We will characterize $\mathbf{Q}^n_{\rm B}(P,P^*)$ by its extreme rays. A ray
with direction vector $x\neq 0$ is a set $\{\lambda x\}$ ($\lambda \geq 0$).
$l$ is an extreme ray of a cone $\mathbf{Q}$ if for any $u \in l$ and any $x,y
\in \mathbf{Q}$, whenever $u = (x + y)/2$, we must have $x,y\in l$. If a closed
convex cone does not include a whole straight line then it is the convex hull
of its extreme rays \cite{Rockafellar1970}.

\begin{lemma}\label{Lemma:PointCone}
The cone $\mathbf{Q}^n_{\rm B}(P,P^*)$ does not include a whole straight
line. \end{lemma}
\begin{proof}If $v\neq 0$ is a possible value of the right hand side of
(\ref{MAsterEq1}) then the derivative of $H_2(P\|P^*)$ in direction $v$ is
strictly negative (\ref{QuadLyap}). Therefore, it is impossible that both $v$
and $-v$ belong to $\mathbf{Q}^n_{\rm B}(P,P^*)$.
\end{proof}

Let us consider a simple cyclic scheme, $A_{i_1} \to \ldots \to A_{i_k} \to
A_{i_1}$ ($k\leq n$ and all the numbers $i_1, \ldots, i_k$ are different). For
a given positive equilibrium, $P^*$ the coefficients for this scheme belong to
a ray:
\begin{equation}\label{CycleRates}
q_{i_{j+1} \, i_{j} }= \frac{\kappa}{p_{i_j}^*} \; (j=1, \ldots , k) \, ,
\end{equation}
where $\kappa \geq 0$ is a constant and we use the standard convention that for
a cycle $q_{i_{k+1} \, i_{k} }=q_{i_{1} \, i_{k} }$.

\begin{lemma}\label{Lemma:CycleExistence}
If system (\ref{MAsterEq0}) has a positive equilibrium $P^*$ then for every
$A_i$ either all $q_{ji}=q_{ij}=0$ or the state $A_i$ belongs to a cycle with
strictly positive rate constants.
\end{lemma}
\begin{proof}
Let $A_i$ not belong to a cycle with positive constants. We say that a state
$A_j$ is reachable from a state $A_k$ if there exists a non-empty chain of
transitions with non-zero coefficients which starts at $A_k$ and ends at $A_j$:
$A_k \to \ldots \to A_j$. Let $\mathcal{A}_{i \downarrow}$ be the set of states
reachable from $A_i$ and $\mathcal{A}_{i \uparrow}$ be the set of states $A_i$
is reachable from. $\mathcal{A}_{i \downarrow} \cap \mathcal{A}_{i \uparrow} =
\emptyset$ because  $A_i$ does not belong to a cycle. If $\mathcal{A}_{i
\uparrow}$ is not empty then in equilibrium all the corresponding $p_j^*=0$ ($j
\in \mathcal{A}_{i \downarrow}$ because there is flow from $\mathcal{A}_{i
\uparrow}$ to $A_i$ and no flow back). If $\mathcal{A}_{i \downarrow}$ is not
empty then in equilibrium $p^*_i=0$ because there is a flow from $A_i$ to
$\mathcal{A}_{i \downarrow}$ and no flow back. Therefore, if the equilibrium is
strictly positive and $A_i$ does not belong to a cycle then $\mathcal{A}_{i
\uparrow} = \mathcal{A}_{i \downarrow}=\emptyset$, hence, all
$q_{ji}=q_{ij}=0$.
\end{proof}

We will use the following simple general statement: Let $\mathcal{Q}$ be a cone
in $\mathbb{R}^m$ without straight lines, $L$ be a linear map, $L:\mathbb{R}^m
\to \mathbb{R}^k$, and $\mathbf{Q}=L(\mathcal{Q})$ be a cone in $\mathbb{R}^k$
without straight lines. Then for every extreme ray $V \subset \mathbf{Q}$ there
exists an extreme ray $W \subset \mathcal{Q}$ such that $L(W)=V$. (In other
words, there always exists an extreme ray in the preimage of an extreme ray.)
We will apply this statement to $\mathcal{Q}=\mathcal{Q}^n_{\rm B}(P^*)$ (the
cone of all Markov processes with the given equilibrium $P^*$) and
$\mathbf{Q}=\mathbf{Q}^n_B{(P,P^*)}$ (the cone of the possible velocities at
point $P$ for all Markov processes with the given equilibrium). The map $L$
transforms the right hand side of the Kolmogorov equation (\ref{MAsterEq0})
into its value at point $P$. This transformation ``vector field $\mapsto$ its
value at point $P$" is, obviously, a linear map.

\begin{lemma}\label{Lemma:ExtremeCycles}Any extreme ray of the cone $\mathcal{Q}^n_{\rm
B}(P^*)$ is a simple cycle with constants (\ref{CycleRates}).
\end{lemma}
\begin{proof}Let a non-zero Markov chain $Q$ with
coefficients $q_{ij}$ belong to an extreme ray of $\mathcal{Q}^n_{\rm B}(P^*)$.
Due to Lemma~\ref{Lemma:CycleExistence} this chain includes a simple cycle with
non-zero coefficients, $A_{i_1} \to \ldots \to A_{i_k} \to A_{i_1}$ ($k\leq n$,
all the numbers $i_1, \ldots, i_k$ are different, $q_{i_{j+1}\,i_{j}}>0$ for
$j=1,\ldots, k$, and $i_{k+1}=i_1$). For sufficiently small $\kappa$
($0<\kappa<\kappa_0$), $q_{i_{j+1}\,i_{j}}-\frac{\kappa}{p^*_{i_{j}}}>0$
($j=1,\ldots, k$). Let $Q_{\kappa}$ be the same simple cycle with the
coefficients (\ref{CycleRates}). Then for $0<\kappa <\kappa_0$ vectors $Q\pm
Q_{\kappa}$ also represent  Markov chains with the equilibrium $P^*$.
Obviously, $Q=\frac{(Q+Q_{\kappa})+(Q-Q_{\kappa})}{2}$, hence,  $Q$ should be
proportional to $Q_{\kappa}$.
\end{proof}

Due to this Lemma, every Markov chain with positive equilibrium is a convex combination
of several simple cycles with the same equilibrium. This is the {\em global
decomposition} of a Markov chain into simple cycles. ``Global" here means that the same
decomposition is valid for all distributions.

{\em Now, we are in position to prove Theorem 1.}
\begin{proof}
We will prove that any extreme ray of the cone $\mathbf{Q}^n_{\rm B}(P,P^*)$ corresponds
to a simple cycle of length 2 (a {\em step}): $A_i \rightleftharpoons A_j$ with the rate
constants (\ref{CycleRates}) $q_{ij} = \frac{\kappa}{p^*_j}\, , \; q_{ji} =
\frac{\kappa}{p^*_ i}$.

According to Lemma~\ref{Lemma:ExtremeCycles}, it is sufficient to prove that
for any simple cycle with equilibrium $P^*$ and rate constants
(\ref{CycleRates}) and for any distribution $P$ the right hand side of the
Kolmogorov equation (\ref{MAsterEq0}) is a conic combination (a combination
with non-negative real coefficients) of the right hand sides of this equation
for simple cycles of length 2 at the same point $P$.

Let us prove this by induction on the cycle length $k$. For $k=2$ it is true
(trivially). For a cycle of length $k>2$, $A_1 \to A_2 \to \ldots A_k \to A_1$,
with the rate constants given by (\ref{CycleRates}), the right hand side of
equation (\ref{MAsterEq0}) is the vector $\mathbf{v}_k$ with coordinates
\begin{equation}\label{cycleVelocity} (\mathbf{v}_k)_j=\frac{p_{j-1}}{p^*_{j-1}}
-\frac{p_{j}}{p^*_{j}}
\end{equation}
Here, without loss of generality, we take $\kappa=1$, use index $j$ instead of
$i_j$ and apply the standard convention regarding cyclic order. Other
coordinates of $\mathbf{v}_k$ are zeros.

Let us decompose this $\mathbf{v}_k$ into a conic combination of a vector
$\mathbf{v}_{k-1}$ for a cycle of length $k-1$ and a vector $\mathbf{v}_2$ for
a cycle of length 2. The flux $A_{j}\to A_{j+1}$ is ${p_{j}}/{p^*_{j}}$. Let us
find the minimum value of this flux and, for convenience, let us put this
minimal flux in the first position by a cyclic permutation. The target cycle of
length $k-1$ is $A_{2} \to \ldots A_{k} \to A_{2}$ with rate constants given by
formula (\ref{CycleRates}) ($\kappa=1$). We just delete the vertex with the
smallest flux from the initial cycle of length $k$. The target cycle of length
2 is $A_1\rightleftharpoons A_2$ with the rate constants (\ref{CycleRates})
$q_{21} = \frac{\kappa}{p^*_1}\, , \; q_{12} = \frac{\kappa}{p^*_ 2}$. We find
the constant $\kappa$ from the conditions:
$\mathbf{v}_k=\mathbf{v}_{k-1}+\mathbf{v}_2$ at the point $P$, hence, two
following reaction schemes, (a) and (b), should have the same velocities, $\D
P/\D t$:

$$\mbox{(a)  }A_k { \overset{1/p_k^*}{\rightarrow}}A_1{ \overset{1/p_1^*}{\rightarrow}}A_2 \mbox{  and  (b)  }
A_k { \overset{1/p_k^*}{\rightarrow}}A_2 \, ; A_1
\underset{\kappa/p_2^*}{\overset{\kappa/p_1^*}{\rightleftharpoons}}A_2 \, .$$
From this condition,
$$\kappa=\left(\frac{p_{k}}{p_{k}^*}-\frac{p_{1}}{p_{1}^*}\right)
\left(\frac{p_{2}}{p_{2}^*}-\frac{p_{1}}{p_{1}^*}\right)^{-1}$$ $\kappa \geq 0$
because ${p_{1}}/{p^*_{1}}$ is the minimal value of ${p_{j}}/{p^*_{j}}$.
Finally,  $\mathbf{v}_k=\mathbf{v}_{k-1}+\mathbf{v}_2$.
\end{proof}

Further, we omit the index B or DB at the cone: $\mathbf{Q}^n_{\rm
B}(P,P^*)=\mathbf{Q}^n_{\rm DB}(P,P^*)=\mathbf{Q}^n (P,P^*)$.

It is necessary to stress that the decomposition of the right hand side of the Kolmogorov
equation (\ref{MAsterEq0}) into a conic combination of cycles of length 2 depends on the
ordering of the ratios $p_i/p_i^*$ and cannot be performed for all values of $P$
simultaneously. Thus, this decomposition is {\em local}.

\subsection{Quasichemical representation and the cones of possible velocities \label{Quasichem}}

For systems with detailed balance, the cone of possible velocities, $\mathbf{Q}^n_{\rm
DB}(P,P^*)$, is a polyhedral cone. For a given $P^*$, it is a piecewise constant function
of $P$. The hyperplanes of the equilibria $A_i\leftrightarrow A_j$ divide the standard
simplex of distributions into a finite number of polyhedra ({\em compartments}). In each
compartment the dominant direction of every transition $A_i\leftrightarrow A_j$ is fixed
and the cone of possible velocities is constant. Now we find that this construction provides the cones
of  possible velocities for general Markov kinetics and not only for systems with detailed balance. Let us describe these cones in detail.

The construction of cones of possible velocities was described in 1979 \cite{Gorban1979}
for systems with detailed balance in the general setting for generalized MAL, for
nonlinear chemical kinetics. These systems are represented by stoichiometric equations of
the elementary reaction coupled with the reverse reactions:
\begin{equation}\label{Stoichiometric}
\alpha_{\rho 1}A_1+\ldots + \alpha_{\rho n}A_n \rightleftharpoons \beta_{\rho 1}A_1+\ldots
+ \beta_{\rho n}A_n\, ,
\end{equation}
where $\alpha_{\rho i}, \, \beta_{\rho i} \geq 0$ are the stoichiometric
coefficient, $\rho $ is the reaction number ($\rho =1, \ldots, m$). The
stoichiometric vector of the $\rho $th reaction is an $n$ dimensional vector
$\gamma_\rho $ with coordinates $\gamma_{\rho i} = \beta_{\rho i}- \alpha_{\rho
i}$.

The equilibria of the $\rho $th pair of reactions (\ref{Stoichiometric}) form a
surface in the space of concentrations. The intersection of these surfaces for
all $\rho $ is the equilibrium (with detailed balance). These surfaces of the
equilibria of the pairs of elementary reactions (\ref{Stoichiometric}) divide
the space of concentrations into several compartments. In each compartment the
dominant direction of each reaction (\ref{Stoichiometric}) is fixed and, hence,
the cone of possible velocities is also constant. It is a piecewise constant
function of concentrations (for a given temperature): $$\mathbf{Q}={\rm
cone}\{\gamma_{\rho } {\rm sign}(w_{\rho }) \, | \, \rho =1, \ldots , m\}\, .$$

For example, let us join the transitions $A_i \rightleftharpoons A_j$ in pairs (say,
$i>j$) and introduce the {\em stoichiometric vectors} $\gamma^{ij}$ with coordinates:
\begin{equation}
\gamma^{ij}_k=\left\{\begin{array}{ll}
-1 &\mbox{ if } k=i,\\
1 &\mbox{ if } k=j,\\
0 &\mbox{ otherwise}.
\end{array}
\right.
\end{equation}
Let us rewrite the Kolmogorov equation for the Markov process with detailed balance
(\ref{detBal}) in the quasichemical form:
\begin{equation}\label{QuasiChemKol}
\frac{\D P}{\D t}=\sum_{i>j}w_{ij}^*\left(\frac{p_j}{p_j^*}-\frac{p_i}{p_i^*}\right) \gamma^{ji}\, .
\end{equation}
Here, $w_{ij}^*=q_{ij}p^*_j=q_{ji}p^*_i$ is the equilibrium flux from $A_i$ to $A_j$ and
reverse.

The cone of possible velocities for (\ref{QuasiChemKol}) is
\begin{equation}
\mathbf{Q}^n (P,P^*)={\rm cone}\left\{\gamma^{ji}
{\rm sign}\left(\frac{p_j}{p_j^*}-\frac{p_i}{p_i^*}\right) \, \Big| \, i>j \right\}\, .
\end{equation}
Here, we use the three-valued sign function (with values $\pm 1$ and 0). In
Fig.~\ref{3stateCones}, the partition of the standard distribution simplex into
compartments, and the cones (angles) of possible velocities are presented for
the Markov chains with three states.

\begin{figure}
\centering{\includegraphics[width=0.5\textwidth]{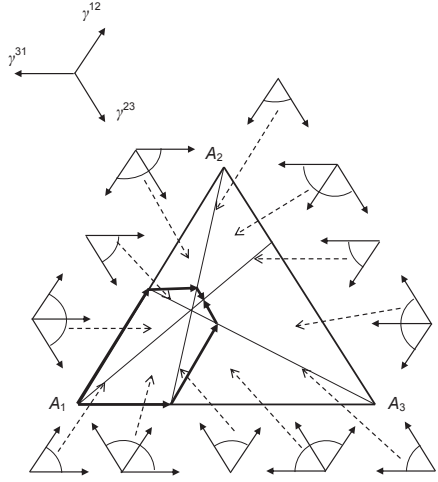}
\caption{\label{3stateCones}Partition of the distribution triangle for the
Markov chains with three states by the lines
$\frac{p_j}{p_j^*}-\frac{p_i}{p_i^*}=0$ into twelve compartments. The
corresponding cones (angles) of possible velocities are presented. The
clockwise and anticlockwise borders of the trajectories are represented by bold
lines.}}
\end{figure}

A set of distributions $U$ is {\em positively invariant} with respect to system
(\ref{MAsterEq0}) if for any initial distribution $P(0)\in U$, the solution of
(\ref{MAsterEq0}) $P(t)$ remains in $U$ for $t>0$. The bold broken lines in
Fig.~\ref{3stateCones} follow along the extreme rays of the angles of possible
velocities (clockwise or anticlockwise). They form the borders of a
positively-invariant area for all the Markov chains with the given equilibrium
$P^*$.

These borders give, for example, a simple estimate of the logarithmic decrement
for Markov chains. For decaying oscillations, the logarithmic decrement is the
natural logarithm of the ratio of any two successive amplitudes: $\delta
\triangleq \ln\frac{x_1}{x_2}$. For a complex eigenvalue $\lambda$, the period
between two amplitudes $T=2\pi/|\Im \lambda|$ and $\delta = 2\pi \frac{|\Re
\lambda|}{|\Im \lambda|}$. For systems with detailed balance, eigenvalues are
always real but for the general Markov chains they may be complex. For example,
for the simple cycle $A_1\to A_2 \to A_3 \to A_1$ with the equilibrium
equidistribution $p_{1,2,3}^*=1/3$, and the rate coefficients $\kappa$, the
nonzero eigenvalues of the linear system (\ref{MAsterEq0}) are $\lambda=\kappa
(-\frac{3}{2}\pm i\frac{\sqrt{3}}{2})$ and $\delta = 2\pi \sqrt{3}$.

Let us follow the clockwise border trajectory (Fig.~\ref{3stateCones}) starting
from the state $A_1$ (the corresponding distribution is $P=(1,0,0)$). This
state belongs to the line of equilibria of the transition $A_2
\rightleftharpoons A_3$. The first step is the equilibration of the transition
$A_1 \rightleftharpoons A_2$ ($A_3$ does not change). After that, the
equilibration of the transition $A_1 \rightleftharpoons A_3$ follows ($A_2$
does not change):
\begin{equation}
\begin{split}
&(1,0,0) \mapsto
\left(\frac{p^*_1}{1-p^*_3},\frac{p^*_2}{1-p^*_3},0\right) \\
 &\mapsto \left(\frac{(p^*_1)^2}{(1-p^*_3)(1-p^*_2)},\frac{p^*_2}{1-p^*_3},\frac{p^*_1
p^*_3 }{(1-p^*_3)(1-p^*_2)}\right)\, .
\end{split}
\end{equation}
As the result of this sequence of equilibrations, when the clockwise border
line  again approaches the equilibrium line of the transition $A_2
\rightleftharpoons A_3$, the value of $p_1$ is
$\frac{(p^*_1)^2}{(1-p^*_3)(1-p^*_2)}$. After this turn in angle $\pi$ every
trajectory becomes closer to $P^*$. The contraction coefficient is
$\frac{p^*_1p^*_2 p^*_3}{(1-p^*_1)(1-p^*_2)(1-p^*_3)}$ or less. The
anticlockwise trajectory gives the same contraction. We estimated the
logarithmic decrement from below:
\begin{equation}\label{DecremEval}
\delta \left(=2\pi \frac{|\Re \lambda|}{|\Im \lambda|}\right) \geq 2\ln
\left(\frac{(1-p^*_1)(1-p^*_2)(1-p^*_3)}{p^*_1p^*_2 p^*_3}\right) \, .
\end{equation}

\subsection{Two $H$-theorems \label{H-theorems}}

The most general form of the $H$-theorem for Markov processes  was proposed by
Morimoto \cite{Morimoto1963}. He used the following $H$-functions: for each
convex function of the positive convex variable $h(x)$ the $h$-{\em divergence}
between distributions $P$ and $P^*$ is
\begin{equation}\label{Morimoto}
\begin{split}
&H_h(P \| P^*)=\sum_i p^*_i h\left(\frac{p_i}{p_i^*}\right)\, .
\end{split}
\end{equation}
At the same time these divergencies were studied by Csisz\'ar
\cite{Csiszar1963} and sometimes they are called the Csisz\'ar--Morimoto
divergences. These functions were introduced two years earlier by R\'enyi on
the last page of his famous work \cite{Renyi1961} together with the hint about
the $H$-theorem. For more details see \cite{GorbanGorbanJudge2010}.

The time derivative of the Csisz\'ar--Morimoto function $H_h(P \| P^*)$
(\ref{Morimoto}) with respect to master equation (\ref{MAsterEq1}) for a
general Markov process is
\begin{equation}\label{ENtropyProd}
\begin{split}
&\frac{\D H_h(P \| P^*)}{\D t}=\sum_{i,j, \, j\neq i} q_{ij}p^*_j \\
&\times \left[h\left(\frac{p_i}{p_i^*}\right)-h\left(\frac{p_j}{p_j^*}\right) +
h'\left(\frac{p_i}{p_i^*}\right)\left(\frac{p_j}{p_j^*}-
\frac{p_i}{p_i^*}\right)\right] \leq 0
\end{split}
\end{equation}
For a Markov process with detailed balance we use the quasichemical form of
master equation (\ref{QuasiChemKol}) and find immediately
\begin{equation}\label{ENtropyProdDB}
\begin{split}
&\frac{\D H_h(P \| P^*)}{\D t}=-\sum_{i,j, \, i>j} q_{ij}p^*_j \\
&\times\left(\frac{p_j}{p_j^*}-\frac{p_i}{p_i^*}\right)\left(h'\left(\frac{p_j}{p_j^*}\right)-h'\left(\frac{p_i}{p_i^*}\right)\right)
\leq 0 \, .
\end{split}
\end{equation}
The inequality for the general Markov processes (\ref{ENtropyProd}) follows
from Jensen's inequality in the differential form, $h'(x)(y-x)\leq h(y)-h(x)$.
It is valid for left and right limits of $h'$ at any point $x >0$. The
inequality for systems with detailed balance (\ref{ENtropyProdDB}) follows from
the monotonicity of $h'$. In full agreement with Corollary~\ref{Corollary1},
the divergences $H_h(P \| P^*)$ (\ref{Morimoto}) are Lyapunov functions for
systems with detailed balance and for all the Markov processes as well.
Theorem~\ref{Theorem1} has an even stronger corollary.

\begin{corollary}For every Markov process $Q$ with positive equilibrium $P^*$
and for a distribution $P\neq P^*$ there exists a Markov process $Q_{\rm DB}$
with the same equilibrium that obeys the detailed balance condition and has the
following property: For every convex function $h$ the time derivative ${\D
H_h(P \| P^*)}/{\D t}$ for $Q$ coincides at point $P$ with the time derivative
of $ H_h(P \| P^*)$ at this point for $Q_{\rm DB}$.
\end{corollary}

\section{Nonlinear kinetics: detailed balance versus semi-detailed balance \label{Nonlin}}

The general equations of MAL without any restriction on the reaction rate constants
demonstrate all types of non-trivial dynamic behavior, from multiple steady states to
strange attractors \cite{Yablonskiiatal1991,Ertl1990}. It is not a surprise because the
MAL systems can approximate with arbitrary accuracy any smooth vector field which
preserves the linear conservation laws and positivity of concentrations
\cite{GorbanByYa1986,Kowalski1993}.

The systems with semi-detailed balance give the direct nonlinear generalization of the
general Markov kinetics. They were introduced by Boltzmann for gas kinetics
\cite{Boltzmann1887} and generalized later for MAL systems
\cite{HornJackson1972,GorbanShahzad2011,SzederkHangos2011}. The systems with semi-detailed
balance are the generalized MAL systems with additional relations between rate constants.

To produce these relations, let us follow the classical work \cite{Stueckelberg1952} and
assume that behind the reaction mechanism (\ref{Stoichiometric}) there is the reaction
mechanism with intermediate {\em compounds} $B_{\rho}^{\pm}$ illustrated by
Fig.~\ref{Compounds}. Each compound is associated with a formal input or output complex
$\sum_i\alpha_{\rho i}A_i$ or $\sum_i \beta_{\rho i} A_i$. Such a complex may participate
in several reactions. Let there be $k$ different vectors among
$\{\boldsymbol{\alpha}_{\rho}, \boldsymbol{\beta}_{\rho}\}$
($\boldsymbol{\alpha}_{\rho}=(\alpha_{\rho i})$ and
$\boldsymbol{\beta}_{\rho}=(\beta_{\rho i})$). We denote these different vectors by
$\boldsymbol{\nu}_j$ ($j=1, \ldots, k$). The correspondent complexes are $\Theta_j=\sum_i
\nu_{ji} A_i$. The reaction mechanism (\ref{Stoichiometric}) takes the form of the list
of transitions $\Theta_j \to \Theta_l$ and the extended reaction mechanism is the list of
transitions
\begin{equation}\label{ExtendedStoi}
\Theta_j \rightleftharpoons B_j \to B_l \rightleftharpoons \Theta_l \, .
\end{equation}

Stueckelberg introduced this representation for the collisions in Boltzmann's
equations and used two asymptotic assumptions:
\begin{itemize}
\item{The compounds $B_j$ are in fast equilibrium with the corresponding input or
    output reagents and the reactions $\Theta_j \rightleftharpoons B_j$ in
    (\ref{ExtendedStoi}) are always close to equilibrium (this is the
    quasiequilibrium assumption, QE);}
\item{They exist in very small concentrations compared to other components (this
    leads to the quasi steady state approximation, QSS).}
\end{itemize}
We call the intermediates $B_j$ compounds following the classical work of
Michaelis and Menten \cite{MichaelisMenten1913}. In 1913, they introduced the
same asymptotic assumptions and representation for an enzyme reaction and
demonstrated that in this case the overall catalytic reaction obeys the MAL.

In more general settings, these two assumption, QE and QSS, allow us to produce
the reaction rates for the rates of the overall reactions in the form of the
generalized MAL. The rate of the reaction
$$\sum_i\alpha_{\rho i}A_i \to \sum_i \beta_{\rho i} A_i$$ is the product of two factors,
a standard Boltzmann factor $W_{\rho}$ and a kinetic factor $\varphi_{\rho} \geq 0$:
\begin{equation}\label{generalizedMAL}
r_{\rho}=\varphi_{\rho}W_{\rho}=\varphi_{\rho} \exp\left(\frac{\sum_i\alpha_{\rho i}
\mu_i}{RT}\right)\, ,
\end{equation}
where $\mu_i$ is the chemical potential  of the component  $A_i$. The
corresponding kinetic equation is
\begin{equation}\label{kinur}
\frac{\D N }{\D t}=V \sum_{\rho}r_{\rho} \boldsymbol{\gamma}_{\rho}\, , \;\;
(\gamma_{\rho i}=\beta_{\rho i} - \alpha_{\rho i})\, .
\end{equation}
Here $N$ is the vector of composition ($N_i$ is the amount of $A_i$), and $V$
is the volume.

We use the notation $c_i$ for the concentration of $A_i$, $c$ is the vector of
concentrations, $\varsigma_j$ is the concentration of $B_j$. The chemical
potentials $\mu_i$ of the components  $A_i$ are the partial derivatives of the
free energy density, $\mu_i=\partial f(c,T)/\partial c_i$. The standard
thermodynamic assumption about strong convexity of the function $f(c,T)$ for
all $T$ is accepted.

Let us demonstrate how the generalized MAL follows from the QE and QSS approximations
(for more details see Ref. \cite{GorbanShahzad2011}). The thermodynamic equilibria for the extended mixture are defined as the
conditional minima of the free energy $F$. The free energy of a mixture of
$A_i$ with small admixtures of the compounds $B_j$ is:
\begin{equation}\label{FreeEn1}
F=Vf(c,T)+VRT \sum_{j=1}^q \varsigma_j \left(\frac{u_j(c,T)}{RT}+\ln
\varsigma_j-1\right)\, .
\end{equation}
The entropic terms $VRT \varsigma_j \ln \varsigma_j$ in this expression corresponds to the
ideal gas equations $p_j=\varsigma_jRT$ for the partial pressure of the small admixtures of the compounds $B_j$.
This ideal gas low may be valid not only in gases but for osmotic pressure of small admixtures in solutions
(the Morse equation).

The thermodynamic equilibria of $B_j$ are: $\varsigma_j^{\rm
eq}=\varsigma_j^*/Z$, where $Z$ is a positive number and $\varsigma_j^*$ is the
{\em standard equilibrium}:
\begin{equation}\label{StandEquili}
\varsigma_j^*(c,T)=\exp\left(-\frac{u_j(c,T)}{RT}\right)\, .
\end{equation}
The thermodynamic equilibrium condition of the reactions $\Theta_j
\rightleftharpoons B_j$ under the condition of smallness of $\varsigma_j$
(QE+QSS) can be solved explicitly:
\begin{equation}\label{equilibrationEq}
\varsigma_j^{\rm qe}=\varsigma^*_j(c,T)\exp\left(\frac{\sum_i \nu_{ji}
\mu_i(c,T)}{RT}\right)\, .
\end{equation}

The smallness of the concentration of the compounds implies that  the rates of the
reactions $B_i \to B_j$ in the extended mechanism (\ref{ExtendedStoi}) are linear
functions of their concentrations. Let the rate constants for this first order kinetics
be $q_{ji}$.

In the selected approximations the extended reaction mechanism (\ref{ExtendedStoi})
returns to the form $\Theta_j \to \Theta_l$. The reaction rate of the transition
$\Theta_j \to \Theta_l$ in the quasiequilibrium approximation is $r_{lj}=q_{lj}
\varsigma_j^{\rm qe}$. This is exactly the generalized MAL (\ref{generalizedMAL}) with
$$\varphi_{lj}=q_{lj} \varsigma_j^*,\, \boldsymbol{\alpha}_{\rho}=\boldsymbol{\nu}_j, \,
\boldsymbol{\beta}_{\rho}=\boldsymbol{\nu}_l \, .$$

At the equilibrium $\varsigma^*/Z$, the first order kinetics of compounds should satisfy the general balance
condition (\ref{MasterEquilibrium}):
$$\sum_j q_{lj} \varsigma_j^*= \sum_j q_{jl} \varsigma_l^* \, .$$
Therefore, the kinetic factors $\varphi_{\rho}$ satisfy the identity of semi-detailed
balance:
\begin{equation}\label{semidetailed}
\sum_{\rho, \,\boldsymbol{\alpha}_{\rho}=\boldsymbol{\nu}} \varphi_{\rho}\equiv
\sum_{\rho, \,\boldsymbol{\beta}_{\rho}=\boldsymbol{\nu}} \varphi_{\rho}
\end{equation}
for any vector $\boldsymbol{\nu}$ from the set of all vectors $\{\boldsymbol{\alpha}_{\rho},
\boldsymbol{\beta}_{\rho}\}$. This identity is exactly the Markov balance condition
(\ref{MasterEquilibrium}) for kinetics of compounds with equilibrium
$\varsigma^*/Z$. It has a very transparent sense: the thermodynamic
equilibrium is, at the same time, the equilibrium for the first
order kinetics of compounds, i.e. the it satisfies the balance
condition (\ref{MasterEquilibrium}) for master equation.

Let us assume that the Markov kinetics of compounds satisfies
the detailed balance condition (\ref{detBal}) at the thermodynamic equilibrium:
$$q_{lj}\varsigma_j^*=q_{jl}\varsigma^*_l \, .$$
Then the kinetic factors $\varphi_{\rho}$ satisfy the condition of detailed balance:
\begin{equation}\label{NlinDetailed}
\varphi_{\rho}^+\equiv \varphi_{\rho}^- \, ,
\end{equation}
where $\varphi_{\rho}^+$ is the kinetic factor for the direct reaction and
$\varphi_{\rho}^-$ is the kinetic factor for the reverse reaction.

This detailed balance condition assumes that the sums in the left and right hand sides of Eq. (\ref{semidetailed})
are equal term by term. Therefore, it is stronger than the semi-detailed balance condition.

For linear systems, the semi-detailed balance condition turns into the standard balance
condition (\ref{MasterEquilibrium}) and the detailed balance condition
(\ref{NlinDetailed}) turns into (\ref{detBal}). Of course, the class of systems with
semi-detailed balance is much wider than the class of systems with detailed balance.
Nevertheless, locally they coincide: {\em for given thermodynamic functions
(\ref{FreeEn1}) and any given concentrations and temperature,  the cone of possible
velocities for systems (\ref{kinur}) with semi-detailed balance coincides with the cone of
the possible velocities for the systems with detailed balance.}

Indeed, for the given values of concentrations we can perform the following three
operations, (i) return from the generalized MAL to the first order kinetic equations of
compounds, (ii) use Theorem~\ref{Theorem1} and find the system of compounds with detailed
balance, which has the same velocity at the same point, and (iii) return back, to the
generalized MAL. As a result, we get a kinetic system for the components $A_i$ with
detailed balance, the same free energy $Vf(c,T)$ (\ref{FreeEn1}), and the same velocity
at the selected values of concentrations.

\section{Conclusion}

The definition of detailed balance includes the rates of all transitions at equilibrium
but observability of all these rates together is a very special situation. Typically, one
can observe the overall system velocity, $\D P/\D t$, or just some components of this
velocity but not the rates of individual transitions. According to our results, if we
know the equilibrium distribution $P^*$ and observe the system velocity at one
nonequilibrium point $P$ then we can never distinguish a general system from the systems
with detailed balance. This is true for Markov kinetics as well as for the systems with
the generalized MAL; detailed balance can never be distinguished from the semi-detailed
balance if we know the equilibrium and observe the velocity at one nonequilibrium point.
The difference between velocities of the general kinetic systems and the systems with
detailed balance is hidden in the correlations between different nonequilibrium states
(or, for example, in the continuous pieces of trajectories). The cone of possible
velocities at a nonequilibrium state $P$ is a piece-wise constant function of $P$, which
can be constructed explicitly for the systems with detailed balance
(Fig.~\ref{3stateCones}), and the same construction is valid for the general kinetics.
These results seem to be rather surprising.

For the nonlinear mass action systems, the systems with semi-detailed balance give the
proper analogue of the general Markov kinetics. The conditions of semi-detailed balance
were invented for the Boltzmann equation by Boltzmann \cite{Boltzmann1887}, studied by
Stueckelberg \cite{Stueckelberg1952} and rediscovered for the mass action kinetics by
Horn and Jackson \cite{HornJackson1972}. Recently \cite{GorbanShahzad2011}, it was proved
that the generalized MAL with the semi-detailed balance condition always follows from the
Markov kinetics of compounds in the Michaelis--Menten--Stueckelberg asymptotic. The class
of the systems with semi-detailed balance is wider than the class of systems with
detailed balance. Nevertheless, for a given equilibrium and for any given value of
concentration these two classes have the same sets of possible velocities in the
distribution space.

\end{document}